\documentclass[conference,letterpaper,onecolumn]{IEEEtran}
  \usepackage[margin=1in]{geometry}
\IEEEoverridecommandlockouts

\usepackage{cite}
\usepackage{amsmath,amssymb,amsfonts,amsthm}
\usepackage{algorithmic}
\usepackage{graphicx}
\usepackage{textcomp}
\usepackage{xcolor}
\usepackage{url}

\newtheorem{example}{Example}
\def\BibTeX{{\rm B\kern-.05em{\sc i\kern-.025em b}\kern-.08em
    T\kern-.1667em\lower.7ex\hbox{E}\kern-.125emX}}

\def\E{ {\mathbb{E}}}

\def\HelstOp{\mathrm{HelstOp}}
\def\HelstM{\mathrm{HelstM}}
\def\Out{\mathrm{Out}}
\def\Packing{\mathcal{P}}
\def\HelstPacking{\mathcal{HP}}
\def\Id{\mathbb{I}}
\def\Tr{\mathrm{Tr}}
\def\CovNum{\mathcal{N}}

\def\Learner{\mathcal{L}}
\def\N{\mathbb{N}}
\def\LearningScenario{\mathcal{S}}
\def\AdmissibleDist{\mathcal{F}}
\def\Outcome{\mathfrak{O}}
\def\Inst{\mathcal{E}}

\global\long\def\11{\mathbbm{1}}

\global\long\def\+{\oplus}

\def\<{\langle}
\def\>{\rangle}

\ifdefined \var 
\renewcommand{\var}{\mathsf{var}}
\else
\newcommand{\var}{\mathsf{var}}
\fi

\ifdefined \abs \else
 \newcommand{\abs}[1]{\lvert#1\rvert}
\fi

\ifdefined \norm \else
 \newcommand{\norm}[1]{\lVert#1\rVert}
\fi

\ifdefined \set 
  \renewcommand{\set}[1]{\left\{#1\right\}}
\else
  \newcommand{\set}[1]{\left\{#1\right\}}
\fi

\DeclareMathOperator*{\argmin}{arg\,min}

\DeclareMathOperator*{\tensor}{\otimes}

\usepackage{xparse}
\usepackage{braket}
\usepackage{verbatim}

\newtheorem{definition}{Definition}

\newtheorem{lemma}{Lemma}
\newtheorem{proposition}{Proposition}

\def\Pr{\mathrm{P}}

\def\Dist{\mathcal{D}}
\def\E{\mathbb{E}}

\def\R{\mathbb{R}}
\def\union{\bigcup}

\def\union{\bigcup}
\def\intersect{\bigcap}

\def\X{\mathcal{X}}

\def\Hyp{\mathrm{Hyp}}
\def\X{\mathcal{X}}

\def\Uniform{\mathrm{Uniform}}

\def\fat{\mathrm{fat}}

\def\Out{\mathrm{Out}}

\def\Hilb{\mathcal{H}}

\def\Density{\mathrm{D}}

\def\POVMs{\mathbb{M}}

\newtheorem{theorem}{Theorem}

\begin{document}

\title{Advantage of Entangled Learning Rules in Quantum Measurement Class Learning\\
\thanks{Magner and Das are funded by NSF-ANR QISE 2545353.  Magner is also funded by NSF CAREER 2338855 and CCF 2212327.  Magner is also affiliated with the AI Plus Institute at UAlbany.}
}

\author{%
  \IEEEauthorblockN{Arka Prabha Das and Abram Magner}
  \IEEEauthorblockA{Department of Computer Science \\
                    University at Albany, State University of New York\\
                    Albany, NY, USA\\
                    Email: \{adas5, amagner\}@albany.edu}
}

\maketitle

\begin{abstract}
    Learning with data in the form of quantum states is of current interest and has led to a variety of 
    problems that boil down to interaction with the available data via quantum measurement and classical post-processing
    of observed classical outcomes.  
    In quantum measurement
    PAC learning, one is given a sequence of unknown, prepared quantum states and classical labels, along with a hypothesis class of candidate measurements.  The
    task is to select a measurement from the hypothesis class that minimizes a fixed notion of error in prediction of the classical labels via measurement of a new state by the selected hypothesis.
    In this work, we consider the advantage of interacting with the given data in the measurement learning framework using learning rules given by measurements that cannot be implemented using local operations and classical communication (LOCC), as opposed to single-copy %
    learning rules.
    We provide a construction showing that there exist learning scenarios wherein
    single-copy learning rules are asymptotically suboptimal compared to optimal ones.
    We then show that
    learning rules based on entangled measurements enjoy at most a polynomial sample complexity advantage over single-copy learning rules in the PAC learning setting (under a natural joint measurability covering assumption).  
\end{abstract}

\section{Introduction}
\label{sec:introduction}

    Statistics and learning with data in the form of unknown quantum states is of current interest and has led to the study of a variety of 
    problems that boil down to interaction with the available data via quantum measurement and classical post-processing
    of observed classical outcomes.
    For instance, in binary hypothesis testing, the problem is to distinguish between 
    two known quantum states by interacting with a large number of copies of one of them or the other.  In quantum measurement class
    PAC learning (introduced in~\cite{heidari2021theoreticalframeworklearningquantum}  and studied subsequently in~\cite{202203AISTATS_PadMag,Magner2024Fat,heidariNewBounds}), one is given a sequence of prepared states (with unknown analytical descriptions) and observed statistically correlated classical labels, along with a hypothesis class of candidate measurements, and
    the task is to select a measurement from the hypothesis class that minimizes a fixed notion of error in prediction of the classical labels via measurement by the selected hypothesis.  This is a quantum generalization of classical hypothesis class learning motivated by, e.g., quantum sensing and quantum phase recognition.
    In both hypothesis testing and measurement class learning, one might attempt to generalize the solutions to their classical counterparts by interacting with each data sample individually via a measurement, receiving a classical outcome, and then aggregating these outcomes to produce an answer. 
    For instance, in the classical learning setup, the empirical risk minimization (ERM) rule is known to be universal, and recent work shows that a variation of this rule called denoised empirical risk minimization (\textsf{DERM}) is universal in the POVM learning framework~\cite{Magner2024Fat}.
    Denoised empirical risk minimization performs a measurement on each data sample and then computes a weighted average of functions of the outcomes.
    Such a scheme, in both cases, can be described by a single measurement that is implementable using \emph{local operations and classical communication} -- that is, each individual measurement is local in the sense that it acts on only the subsystem of its corresponding data sample, and then the individual measurement outcomes can be thought of as being sent via classical communication to an entity that aggregates them.  Such measurements are thus sometimes called \emph{LOCC measurements}.  We adopt the more common terminology of \emph{separable} or \emph{single-copy} measurements.
    However, in contrast to the classical case, in the quantum versions of the hypothesis testing and learning problems, a broader class of measurements on the tensor product of Hilbert spaces in which the data set lives is available to us. 
    In this work, we study the sample complexity advantage of interacting with the given data in these frameworks using entangled (i.e., non-separable or multi-copy) measurements.
    
    We give a construction of a class of learning scenarios wherein
    single-copy learning rules are asymptotically suboptimal compared to optimal ones.
    We then show  that
    learning rules based on entangled measurements enjoy at most a polynomial sample complexity advantage over single-copy learning rules in the PAC learning setting, 
    under a natural constraint relating to coverability of the measurement class by approximately jointly measurable subsets.

\subsection{Prior work}
\label{sec:prior-work}
Here we review relevant prior work.  Basics of quantum information theory can be found in~\cite{Nielsen_Chuang_2010,WildeBook}.
\paragraph{Work relevant to the learning framework that we study}

    The POVM PAC learning framework with which we will be primarily concerned was introduced in~\cite{heidari2021theoreticalframeworklearningquantum}.  A sample complexity upper bound for learning was introduced in~\cite{202203AISTATS_PadMag}.  A necessary and sufficient condition for learnability, along with a stronger sample complexity bound, was given in~\cite{Magner2024Fat}.  The paper \cite{heidariNewBounds} provides a sample complexity upper bound for convex measurement classes.  Computational aspects have recently been studied in, e.g.,~\cite{heidari2025hadamard}.  All of these works exclusively analyze upper bounds (except for~\cite{Magner2024Fat}), and only for specific single-copy learning rules.

\paragraph{Other quantum learning and statistical frameworks}
    Many other quantum learning frameworks that differ from measurement class learning exist and are well-studied.  These frameworks study \emph{classical} learning problems that use quantum resources (so the quantum states are analytically known) or, alternatively, learning scenarios where data are quantumly encoded, but hypotheses are classical boolean functions~\cite{bshouty1998learning,arunachalam2017optimalquantumsamplecomplexity,arunachalam2017surveyquantumlearningtheory,2020arXiv200208240A}.  In contrast, the measurement class learning problem that we continue to investigate here involves data in the form of quantum systems with an unknown analytical description and hypothesis classes in the form of quantum measurements.

    Distinct from learning, there is a class of statistical tasks called \emph{quantum tomography}.  For instance, state tomography is the following problem: given a dataset consisting of a collection of systems prepared in the same unknown state $\rho$, measure the dataset and output an estimated density matrix $\hat{\rho}$
    that is close to $\rho$ in measure (e.g., fidelity or trace distance)~\cite{SampleOptimalStateTomography,ODonnell2015EfficientQT}.
    Readers may confuse the quantum measurement learning problem with \emph{process tomography}~\cite{ProcessTomography,Kunjummen2021ShadowPT}, which has the goal of estimation of an unknown quantum channel in some metric (e.g., diamond norm).
    In contrast, the quantum measurement learning framework includes as a crucial component a hypothesis class (e.g., a class of variational quantum circuits with a given structure), from which learned measurements are drawn, and a loss function.  Both of these aspects, intuitively, build modeling assumptions into the learning problem, which are not present in tomography.  This has the advantage of reducing
    sample complexity as a function of notions of model complexity (e.g., fat shattering dimension).

\paragraph{Work that studies the advantage of entangled measurements in various settings}
    Several works study the advantage of entangled measurements for statistical tasks.  For instance, in~\cite{NEURIPS2023_acb7ce5a},
    in the context of quantum learning of boolean function classes, it is demonstrated that entangled measurements enjoy at most a polynomial sample complexity advantage over separable measurements.  Entanglement offers a strict advantage in many other scenarios, including, e.g., quantum property testing~\cite{entanglementisnecessary}
    and quantum hypothesis testing~\cite{QuantumSteinsLemma,TowardsQuantumUniversalHypothesisTesting}.
    Most relevant to us is recent work~\cite{lowe2025lowerboundslearningquantum},
    which gives lower bounds for state tomography using single-copy measurements.

\subsection{Our contributions}
\label{sec:contributions}

Our contributions are as follows:
\begin{enumerate}
    \item 
        In the POVM hypothesis class PAC learning framework, we show that there exist learning scenarios such that the sample complexity of learning via single-copy measurements is asymptotically larger (as the Hilbert space dimension tends to $\infty$) than an optimal learning rule.
    \item 
        We show that the sample complexity of learning via single-copy measurements is at most polynomially worse than that of learning via general measurements under natural conditions on the hypothesis class (e.g., that its convex hull has finitely many extreme points).
\end{enumerate}

Many proof details are given in Sections~\ref{proof:separable-learning-rules-suboptimal} and ~\ref{proof:separable-non-separable-gap-bound}, but complete proofs of auxiliary lemmas are deferred to the appendix.  

\section{Main results}
\label{sec:main-results}

\subsection{Preliminaries}
\label{sec:problem-formulation}
In this section, we cover basic definitions from quantum information theory.  We then formulate the statistical frameworks that we study, along with relevant results that we will use.

\subsubsection{Definitions and notation from quantum information theory}
\label{sec:quantum-info-defs}
Fix a Hilbert space $\Hilb$ with dimension $d$, which maybe $\infty$.
Quantum states are analytically described by \emph{density operators} $\rho$, which are positive semidefinite Hermitian operators on $\Hilb$ satisfying $\Tr(\rho) = 1$.  We denote the set of density operators on $\Hilb$ by
$\Density(\Hilb)$.  We fix an arbitrary orthonormal basis $\{ \ket{e_j}\}_{j=1}^{d}$ as the computational basis.

Quantum measurements (in the context where we do not care about post-measurement states) are described via the positive operator-valued measure (POVM) framework.  A POVM with classical outcomes taking values in a (finite or countable) set $\Omega$ is given by a collection of
positive semidefinite operators $\{\Pi_{\omega}\}_{\omega \in \Omega}$ satisfying
$\sum_{\omega \in \Omega} \Pi_{\omega} = \Id$, where $\Id$ is the identity operator on $\Hilb$.  The outcome of measuring a state $\rho$
with a POVM $\Pi$ is a random variable denoted by $\Out(\Pi, \rho)$, whose probability distribution is given by the \emph{Born rule}: for each $\omega \in \Omega$,
$ %
    \Pr[ \Out(\Pi, \rho) = \omega ]
    = \Tr(\rho \cdot \Pi_{\omega}).  
$ %
We denote the set of all POVMs by $\POVMs(\Hilb)$.  A 
\emph{quantum instrument} describes both outcome statistics and post-measurement states and is described by a collection $\Inst$ of trace-non-increasing maps $\{\Inst_\omega\}_{\omega \in \Omega}$ satisfying
$\Tr(\sum_{\omega \in \Omega} \Inst_\omega(\rho)) = \Tr(\rho)$.  For each $\omega \in \Omega$, $\Pr[\Out(\Inst, \rho) = \omega] = \Tr(\Inst_{\omega}(\rho))$, and the post-measurement state when the outcome is $\omega$ is given by $\rho_{\omega} = \frac{\Inst_{\omega}(\rho)}{\Tr(\Inst_{\omega}(\rho))}$.
Our results are concerned with the advantage of different types of quantum instruments as learning rules for measurement class learning.  
We recall the formulations of these types of measurements below.

\begin{definition}[Single-copy instrument/measurement]
    \label{def:locc-measurement}
    \label{def:single-copy-measurement}
    We say that an instrument or measurement $\Pi$ on a tensor product of Hilbert spaces $\Hilb := \Hilb_1\tensor \cdots \tensor \Hilb_m$ is a single-copy (or separable, or unentangled) measurement if it can be implemented via classical post-processing of
    the outcomes of $m$ measurements on the component Hilbert spaces 
    $\{\Hilb_j\}_{j=1}^m$.
    We call an instrument or measurement that is not single-copy \emph{entangled} or \emph{non-separable}.
\end{definition}

\paragraph{Distance measures}
For $\sigma_1, \sigma_2 \in \Density(\Hilb)$, we denote by $\|\sigma_1 - \sigma_2\|_1$ the trace distance between $\sigma_1, \sigma_2$.

For two quantum channels $M_1, M_2$, we denote by $\|M_1 - M_2\|_{\diamond}$ their
diamond distance.  We recall~\cite{WildeBook} that when $M_1$ and $M_2$ are two quantum measurements (e.g., given by POVMs), this distance simplifies to
    $ %
        \|M_1 - M_2 \|_{\diamond}
        = \max_{\kappa \in \Density(\Hilb)} d_{TV}(\Out(M_1, \kappa), \Out(M_2, \kappa)).
    $ %

\subsubsection{POVM hypothesis class PAC learning}
We next describe the POVM hypothesis class learning framework.  We fix a subset $\X$ of the space of density matrices $\Density(\Hilb)$
on a Hilbert space $\Hilb$.  We denote by $\Omega := \X \times \{0, 1\}$, where $\{0, 1\}$ is the set of \emph{class labels}.  
We also fix a \emph{hypothesis class} $\Hyp$ consisting of a set of two-outcome POVMs on $\Hilb$.  We fix a set of \emph{admissible data-generating distributions} $\AdmissibleDist$ with support on $\Omega$.
Each $\Dist \in \AdmissibleDist$ can without loss of generality be parameterized by a 
class-$0$ prior probability $p_0 \in (0, 1)$ and a
fixed 
pair of density operators $\rho$ and $\sigma$, corresponding respectively to the class labels $0$ and $1$.  In particular, a sample $(\kappa, 0)$ or $(\gamma, 1)$ from $\Dist$, where $\kappa$ and $\gamma$ are \emph{randomly chosen}, is 
indistinguishable via measurement from $(\rho, 0)$ or $(\sigma, 1)$, 
respectively.

We collect these ingredients into the following definition.
\begin{definition}[Learning scenario]
    \label{def:learning-scenario}
    A \emph{learning scenario} on $\Hilb$ is a tuple 
    $\LearningScenario = (\Hyp, \AdmissibleDist)$ where $\Hyp$ is a hypothesis class of two-outcome POVMs on $\Hilb$,  $\AdmissibleDist$ is a family of admissible distributions over
    $\X \times \{0,1\}$. Each
    $\Dist \in \AdmissibleDist$ is parameterized by a pair
    $(\rho,\sigma) \in \mathcal{D}(\mathcal{H})^2$ (corresponding to the class-$0$ and class-$1$ density operators) and a label prior
    $p_0 \in (0,1)$. 
\end{definition}

For any hypothesis class $\Hyp$, we define the loss function $\ell:\Hyp\times\Omega \to [0, 1]$ as follows:
$ %
    \ell(\Pi, (\rho, Y))
    := I[ \Out(\Pi, \rho) \neq Y ].
$ %
We note that this is a random variable.  We then define, for a hypothesis $h \in \Hyp$, its risk with respect to the data-generating distribution $\Dist \in \AdmissibleDist$:
$%
    R_{\Dist}(h)
    := \E_{Z \sim \Dist}[ \ell(h, Z) ].
$%
When $\Dist$ is clear from the context, we omit it from the notation.

We next define the PAC POVM learning game associated with a given learning scenario.  It is played between \emph{nature} and a \emph{learner} in a sequence of rounds.
\begin{center}
\fbox{
\begin{minipage}{0.92\linewidth}
\textbf{The PAC POVM Learning Game for $(\Hyp, \AdmissibleDist)$}
\begin{enumerate}%
    \item Nature fixes an unknown distribution $\Dist\in\AdmissibleDist$.
    \item 
        The learner receives $m$ i.i.d. labeled samples
        $%
        S=((\rho_j,Y_j))_{j=1}^m\sim \Dist^m.
        $ %
        Here, we call $S$ the \emph{training set} and the $Y_j$ \emph{class labels}.  We define $\rho^{(0)} := \tensor_{j=1}^m \rho_j$.
    \item 
        In round $t \in \{0, 1, ...\}$, the learner chooses a quantum instrument $\Gamma^{(t)}$
        (possibly depending on previous measurement outcomes and all class labels) to the 
        quantum state
        $\rho^{(t)}$, receiving a sample $\Outcome_t \sim \Out(\Gamma^{(t)}, \rho^{(t)})$.  Then $\rho^{(t+1)}$ is the (unobserved) post-measurement state.
    \item 
        The learner decides to stop after $T < \infty$ rounds, where $T$ depends only on all received outcomes
        $\{\Outcome_t\}_{t\geq 0}$ and class labels.  It then 
        outputs a hypothesis $h \in\Hyp$ based on this information.
\end{enumerate}
\end{minipage}
}
\end{center}

We emphasize that the output of a learner is a hypothesis that is independent of $\tensor_{j=1}^m \rho_j$, conditioned on the measurement outcomes and the class labels.

We say that the learner $\Learner$ is $(\epsilon, \delta)$-probably 
approximately correct (PAC) for the learning scenario $(\Hyp, \AdmissibleDist)$ if there exists a sample size 
$m_0(\epsilon, \delta) \in \N$ such that, for all distributions $\Dist \in \AdmissibleDist$, 
with probability $\geq 1-\delta$ over a sample dataset $S$ of size 
$m \geq m_0(\epsilon, \delta)$ from $\Dist$,
$%
    R_{\Dist}(A(S)) 
    \leq \epsilon + \inf_{h \in \Hyp} R(h).
$ %
The minimum $m_0(\epsilon, \delta)$ for which this holds is called the \emph{sample complexity} of $\Learner$ for the scenario. 
The learner is said to be PAC if it is $(\epsilon, \delta)$-PAC for all $\epsilon, \delta > 0$.

\noindent\textbf{Restrictions on measurements made by learners: }
    We say that $\Learner$ is \emph{non-adaptive} if, for all $t$, $\Gamma^{(t)}$ 
    does not depend on $\{\Outcome_{t'}\}_{t' < t}$.
    Throughout, we will implicitly only consider non-adaptive learning rules.
    We say that a learner $\Learner$ is \emph{single-copy} if each
    $\Gamma^{(t)}$ is a single-copy measurement (which implies that each $\rho^{(t)}$ is separable).  %

\subsection{Main result: Suboptimality of single-copy measurements for PAC learning}
\label{sec:pac-suboptimality}

Our first theorem shows that, for optimality, we \emph{must} sometimes consider
entangled learning rules.

\begin{theorem}[Suboptimality of single-copy learning rules]
    \label{thm:separable-learning-rules-suboptimal}
    There exist learning scenarios 
    for which every single-copy 
    learning rule is asymptotically suboptimal as a function of Hilbert space dimension, for small enough $\epsilon$: every single-copy learning rule has 
    sample complexity $\Omega(d^5)$,
    while an optimal (entangled) learning rule has sample complexity 
    $O(d^4)$, for all $\epsilon < c/d$ for a sufficiently small constant $c > 0$ and $\delta = 1-\Omega(1)$ as $d\to\infty$.
\end{theorem}
\emph{Remark: } The range of $\epsilon$ is a consequence of our Lemma~\ref{lemma:diamond-distance-risk-bound}, relating the excess risk of a hypothesis to its diamond distance to the optimal one.

\subsection{Main result: Upper bound on the gap between single-copy and entangled rules}
\label{sec:separable-non-separable-gap-bound}

Our next result gives a bound on the suboptimality gap for single-copy learning rules.
It says that, so long as the hypothesis class can
be partitioned into a small number of \emph{approximately jointly measurable} subsets,
the sample complexity of single-copy learning only suffers a polynomial disadvantage compared to
an optimal learning rule.  We state the result in terms of the \emph{joint measurability covering number} $\CovNum_{JM}(\Hyp, \epsilon)$ defined in~\cite{Magner2024Fat}, Definition~11.

\begin{definition}[Joint measurability covering number~\cite{Magner2024Fat}]
    \label{def:jm-covering-number}
    For $\epsilon>0$, we say that a subset $H \subseteq \Hyp$
    is $\epsilon$-approximately jointly measurable if there exists
    a single ``root'' measurement $\Gamma$ such that every measurement $\Pi \in \Hyp_i$  can be approximated to diamond distance $\epsilon$ by classical post-processing of the outcome of $\Gamma_i$.
    Then the joint measurability covering number
    $N_{\mathrm{JM}}(\Hyp,\epsilon)$ is the smallest integer $N$ such
    that $\Hyp$ can be covered by subsets
    $ %
        \Hyp\subseteq \union_{i=1}^N \Hyp_i
    $ %
    where each $\Hyp_i$ is $\epsilon$-approximately jointly measurable.
\end{definition}

\begin{example}[Convex class with finitely many extreme points]
    \label{ex:convex}
    Consider a finite-dimensional Hilbert space, and let $\mathrm{Hyp}$ be a convex set of two-outcome POVMs with at most
    $n_e$ extreme points. 
    Then $N_{\mathrm{JM}}(\mathrm{Hyp},\epsilon)
    \leq \lceil \frac{C}{\epsilon} \rceil^{n_e}$. For $n_e = O(1)$, this is
    polynomial in $1/\epsilon$.

    The claimed upper bound on the joint measurability covering number in this setting follows by the known upper bound by the
    diamond distance covering number (\cite{Magner2024Fat}, Lemma~12), 
    then applying known metric entropy bounds for polytopes.
\end{example}

\begin{theorem}[Upper bound on the gap between single-copy and entangled learning rules]
    \label{thm:separable-non-separable-gap-bound}
    There exists a single-copy learning rule $\Learner$ such that,  for any measurement hypothesis class $\Hyp$, if $\Hyp$ is $(\epsilon, \delta)$-PAC learnable
    with $m_0(\epsilon, \delta)$ samples, then $\Learner$ is $(\epsilon, \delta)$-PAC for
    $\Hyp$ with sample complexity at most polynomial in $m_0(\epsilon, \delta), \epsilon, \delta,$ and
    $\CovNum_{JM}(\Hyp, \epsilon)$.
\end{theorem}

\section{Conclusions and discussion}
\label{sec:conclusions}

In this work, we studied the advantage of measurements that exploit non-locality in the quantum statistical framework of POVM class PAC learning.  This question is motivated, on one hand, by the desire for POVM learning rules with optimal sample complexity, and on the other by the cost of implementation of entangled learning rules.  A similar question has been asked (and answered) in a variety of quantum statistical contexts, including quantum property testing and quantum hypothesis testing.

We made progress in the study of POVM learning by exhibiting learning scenarios and accuracy parameter ranges where single-copy learning rules are strictly suboptimal.  We additionally showed that single-copy learning rules suffer at most a polynomial sample complexity detriment over entangled learning rules in general, under a natural assumption entailing coverability of the POVM class with polynomially many almost jointly measurable subsets (as a function of the learning accuracy parameter $\epsilon$).  This assumption is satisfied, for example, whenever the hypothesis class has a convex hull having finitely many extreme points.  

It would be of interest to exhibit practically useful hypothesis classes for which an entangled learning rule has a strict advantage.

\section{Proof of Theorem~\ref{thm:separable-learning-rules-suboptimal} (sketch)}
\label{proof:separable-learning-rules-suboptimal}

We will restrict our attention to a fixed Hilbert space $\Hilb$ with finite dimension $d$ and the hypothesis class $\Hyp$ consisting of all $2$-outcome POVMs.  We also set the class-$0$ density operator to be $\rho := \frac{1}{d}\Id$.  The collection of class-$1$ density operators that we consider will be $\sigma \in \Packing(\gamma)$
defined in Definition~\ref{def:packing-definition} below.  With these choices, we know from quantum state discrimination~\cite{Nielsen_Chuang_2010} that the minimum-risk hypothesis is the H\"elstrom measurement corresponding to $(\rho, \sigma)$ (see Definition~\ref{def:helstrom-measurement} below), which we denote by $\HelstM(\rho,\sigma)$.

We remark that, despite borrowing tools from a hypothesis testing 
problem, the learning scenario that we construct is distinct from 
testing: the goal of learning in this scenario remains to produce a 
hypothesis that performs close to optimally when given a single new 
sample state (without its classical label).  We are not given 
\emph{both} states $\rho$ and $\sigma$ in this setting.  In 
contrast, in the testing problem, one knows the analytical 
descriptions of both $\rho$ and $\sigma$ and must accurately decide 
whether a given register is prepared in state $\rho$ or $\sigma$.

Our proof consists of a sample complexity upper bound
using an entangled learning rule (Proposition~\ref{prop:sample-complexity-upper-bound}), followed by a sample complexity \emph{lower} bound (Proposition~\ref{prop:single-copy-lower-bound}) for \emph{every} single-copy learning rule for a specific learning scenario.  The upper bound is asymptotically smaller than the lower bound as $d\to\infty$, which is the crux of the proof.

\subsection{Preliminaries}

We start by recalling the H\"elstrom measurement from quantum state discrimination~\cite{Nielsen_Chuang_2010}.

\begin{definition}[H\"elstrom measurement of a pair of density operators]
    \label{def:helstrom-measurement}
    Let $\kappa_1, \kappa_2 \in \Density(\Hilb)$.
    We define the H\"elstrom operator by $\HelstOp(\kappa_1, \kappa_2) := \kappa_1 - \kappa_2$.
    Supposing that $\HelstOp(\kappa_1, \kappa_2)$ diagonalizes as 
    $%
        \HelstOp(\kappa_1, \kappa_2) = \sum_{j=1}^d \lambda_j \ket{v_j}\bra{v_j}.
    $%
    We define orthogonal projectors onto the positive and negative eigenspaces
    of $\HelstOp(\kappa_1, \kappa_2)$ as follows:
    $%
        P_{+}(\kappa_1, \kappa_2) := \sum_{j ~:~ \lambda_j > 0} \ket{v_j}\bra{v_j}$ and %
        $P_-(\kappa_1, \kappa_2) := \sum_{j ~:~ \lambda_j < 0} \ket{v_j}\bra{v_j}.
    $ %
    Finally, we define the H\"elstrom measurement to be the following POVM:
    $\HelstM(\kappa_1, \kappa_2) := ( P_-(\kappa_1, \kappa_2), P_+(\kappa_1, \kappa_2))$.
\end{definition}

For our learning sample complexity lower bound, we will recall a packing of states
used in sample complexity lower bound arguments for various quantum statistical tasks,
such as state tomography and quantum property testing~\cite{entanglementisnecessary,lowe2025lowerboundslearningquantum}.
\begin{definition}[Random state packing]
    \label{def:packing-definition}
    Fix $\gamma > 0$.
    We define the set $\Packing(\gamma)$ of states as follows.
    We define
    \begin{align}
        X := \sum_{j=1}^{d/2} \gamma \ket{e_j}\bra{e_j} + \sum_{j=d/2+1}^d (-\gamma)\ket{e_j}\bra{e_j},
    \end{align}
    where we have assumed without loss of generality that $d$ is even.
    We then define
    $ %
        \Lambda := \frac{1}{d}(\Id + X).
    $ %
    We define a sequence of Haar-random unitary operators $\{U_k\}_{k=1}^K$, for
    $K$ to be chosen later.
    Finally, we define
    $ %
        \sigma_{k} := U_k \Lambda U^{\dagger}
    $%
    and
    $ %
        \Packing(\gamma) := \{\sigma_k\}_{k=1}^K.
    $%
\end{definition}

We will call the problem of $K$-ary hypothesis testing, where the hypotheses
are each of the states in the packing, the $\Packing(\gamma)$-hypothesis testing problem.
In~\cite{lowe2025lowerboundslearningquantum}, this problem
is used in the development of lower bounds for
single-copy state tomography.  
They show that with probability at least
$1 - \exp(-\Omega(d^2))$, any given pair of distinct elements $\sigma_j, \sigma_k \in \Packing(\gamma)$ are $\gamma$-separated in trace distance.  Thus, choosing $K = \exp(\Theta(d^2))$,  $\Packing(\gamma)$ is a $\gamma$-packing in trace distance.

Furthermore, ~\cite{lowe2025lowerboundslearningquantum} shows that the sample complexity of distinguishing
among the states in $\Packing(\gamma)$ using single-copy measurements is
$\Omega(d^3/\gamma^2)$, which is then used to conclude their state tomography lower bound in the case of single-copy measurements.

\paragraph{Expressions for the H\"elstrom projectors induced by the packing} %
    In what follows, we will need explicit expressions for the H\"elstrom projectors
    induced by the packing.  Each $U_k$ can be expressed as 
    $ %
        \sum_{j=1}^d \ket{v_j^{(k)}}\bra{v_j^{(k)}}
    $ %
    for an orthonormal basis $\{ \ket{v_j^{(k)}} \}_{j=1}^d$.
    Then we have, by simple algebra,
    \begin{align}
        P_+(\rho, \sigma_k)
        &= \sum_{j=d/2+1}^d  \ket{v_j^{(k)}}\bra{v_j^{(k)}}  \\
        P_-(\rho, \sigma_k)
        &= \sum_{j=1}^{d/2} \ket{v_j^{(k)}}\bra{v_j^{(k)}}.
        \label{expr:helstrom-projector-formulas}
    \end{align}

\paragraph{Auxiliary lemmas}

\begin{lemma}[Excess risk controls diamond distance to $\HelstM(\rho,\sigma)$]
    \label{lemma:diamond-distance-risk-bound}
    Consider the hypothesis class $\Hyp$ consisting of all $2$-outcome POVMs on
    a fixed Hilbert space $\Hilb$.  We have the following upper bound on the
    diamond distance between a hypothesis $h$ and $\HelstM(\rho, \sigma)$, the minimum-risk
    hypothesis:
    $ %
        R(h) - R(\HelstM(\rho, \sigma))
        \geq \frac{\alpha}{4}\| h - \HelstM(\rho, \sigma)\|_\diamond^2,
    $ %
    where we define $\alpha$ to be the smallest singular value of $\rho - \sigma$
    (hence, $\alpha = \gamma/d$).
\end{lemma}

The next lemma allows us to conclude that a good estimate $\hat{\sigma}$ of $\sigma$ in trace distance induces a low-risk hypothesis $h_*(\rho, \hat{\sigma})$.  This lemma
is useful in translating a tomography sample complexity upper bound to a sample 
complexity upper bound for learning.  
\begin{lemma}[Trace distance controls excess risk]
    \label{lemma:trace-distance-risk-bound}
    Let $\hat{\sigma} \in \Density(\Hilb)$.
    We have
    $%
        R(\HelstM(\rho, \hat{\sigma})) - R(\HelstM(\rho, \sigma))
        \leq \frac{1}{2} \|\sigma - \hat{\sigma}\|_1.
    $%
\end{lemma}
Proofs of these lemmas are deferred to the appendix.

\subsection{Sample complexity upper bound via state tomography}

\begin{proposition}[Sample complexity upper bound for $\Hyp$]
    \label{prop:sample-complexity-upper-bound}
    For constant $\delta > 0$, the optimal sample complexity of learning $\Hyp$, over all learning rules, is $O(\frac{d^2}{\epsilon^2})$.
\end{proposition}
\begin{proof}
    We analyze the sample complexity of the following learning rule:
    \begin{enumerate}
        \item
            Using an optimal state tomography protocol, produce an estimate
            $\hat{\sigma}$ of $\sigma$ satisfying
            $ %
                \|\sigma - \hat{\sigma} \|_1 \leq \epsilon.
            $ %
        \item
            Output the H\"elstrom measurement $\HelstM(\rho, \hat{\sigma})$.
    \end{enumerate}

    By Lemma~\ref{lemma:trace-distance-risk-bound}, this learning rule outputs a 
    hypothesis with excess risk at most $\epsilon/2$ with high probability.  The sample 
    complexity depends on that of state tomography in trace distance, which is $O(\frac{d^2}{\epsilon^2})$~\cite{lowe2025lowerboundslearningquantum},
    which completes the proof.
\end{proof}

\subsection{Lower bound for single-copy learning rules}
\begin{proposition}[Single-copy sample complexity lower bound for $\Hyp$]
    \label{prop:single-copy-lower-bound}
    The sample complexity of $(\epsilon, \delta=1-\Omega(1))$-PAC learning $\Hyp$ via any single-copy learning rule with $\epsilon = \frac{\gamma}{cd}$ for large enough constant $c > 0$ is
    $\Omega(d^5)$.
\end{proposition}

To prove Proposition~\ref{prop:single-copy-lower-bound}, we need the following key lemma.
It tells us that the state packing described in Definition~\ref{def:packing-definition} induces a collection of hypotheses
that are nearly maximally separated in diamond distance.
To state it, we define the set $\HelstPacking(\gamma)$ of hypotheses
induced by states in $\Packing(\gamma)$:
$ %
    \HelstPacking(\gamma)
    := \{ \HelstM(\rho, \sigma_j) \}_{j=1}^K.
$ %

\begin{lemma}[Random state packing induces large diamond distance hypothesis packing]
    \label{lemma:state-packing-diamond-distance}
    There exists a constant $const > 0$ such that with probability $1-o(1)$ as $d\to\infty$, for every pair
    of hypotheses $h_j \neq h_k$ in $\HelstPacking(\gamma)$, we have
    $ %
        \|h_j - h_k\|_{\diamond} \geq const.
    $ %
\end{lemma}

To prove Lemma~\ref{lemma:state-packing-diamond-distance}, we need the following
concentration result for Lipschitz functions on the unit sphere.
\begin{lemma}[Concentration for Lipschitz functions on the unit sphere~(\cite{vershynin2018high}, Ch. 5)]
    \label{lemma:lipschitz-sphere-concentration}
    Let $g:S^{n-1}\to\R$ have Euclidean Lipschitz constant $L$ and let $v \sim \Uniform(S^{n-1})$ (here, $S^{n-1}$ is the $n-1$-dimensional unit sphere).  Then for all $t > 0$,
    \begin{align}
        \Pr[ | g(v) - \E[g(v)] | \geq t ]
        \leq 2 \exp \left( \frac{-C n t^2}{L^2} \right).
    \end{align}
\end{lemma}

\begin{proof}[Proof of Lemma~\ref{lemma:state-packing-diamond-distance}]
        Let $\sigma_1, \sigma_2 \in \Packing(\gamma)$.
        Since the $U_j$ are iid, %
        it is sufficient to prove that 
        \begin{align}
            \label{expr:pairwise-diamond-probability}
            \Pr[ \| \HelstM(\rho, \sigma_1) - \HelstM(\rho, \sigma_2)  \|_{\diamond}  < const ] \leq \exp(-\Omega(d^2)).
        \end{align}
        Then, by the choice of $K$ and the union bound over pairs of elements in $\HelstPacking(\gamma)$, we must have the claimed
        result.

       Toward (\ref{expr:pairwise-diamond-probability}), we invoke the definition of
       the diamond distance %
       and the definition of
       measurement by a POVM to conclude that to lower bound the diamond distance,
       it is sufficient to find a state $\kappa \in \Density(\Hilb)$ for which
       $%
           | \Tr(\kappa P_-(\rho, \sigma_1)) - \Tr(\kappa P_-(\rho, \sigma_2))|
           > const.
       $ %

        Recalling the H\"elstrom projector formulas (\ref{expr:helstrom-projector-formulas}), we claim that there exists an element 
        $\kappa \in \{\ket{v_j^{(2)}}\bra{v_j^{(2)}}\}_{j=1}^{d/2}$ for which
        \begin{align}
            | \Tr(\kappa P_-(\rho, \sigma_1)) - \Tr(\kappa P_-(\rho, \sigma_2))|
            \geq 1/4,
            \label{expr:target-trace-difference}
        \end{align}
        with probability $\geq 1 - 2\exp(-\Omega(d^2))$.

        To verify this, note first that for any choice of $j$, $\Tr( \ket{v_j^{(2)}}\bra{v_j^{(2)}} P_-(\rho, \sigma_2)) = 1$.

        It remains to show that with probability $\geq 1 - 2\exp(-\Omega(d^2))$,
        there is some $j$ for which 
        $\Tr( \ket{v_j^{(2)}}\bra{v_j^{(2)}} P_-(\rho, \sigma_1)) < 3/4$.

        To do this, we will use the result on the concentration of Lipschitz functions on the unit sphere.  We define the function
        $g(\ket{v}) := \Tr(\ket{v}\bra{v} P_-(\rho, \sigma_1))$.
        By direct calculation, $g(\ket{v})$ has a Lipschitz constant $O(1)$ as $d\to\infty$.

        Conditioned on $\{\ket{v_k^{(2)}}\}_{k=1}^{j-1}$, $\ket{v_j^{(2)}}$ is
        a uniformly random unit vector in the orthogonal complement of 
        $\{\ket{v_k^{(2)}}\}_{k=1}^{j-1}$.  This implies
        $ %
            \E[ g(\ket{v_j^{(2)}})  ~|~ \{\ket{v_k^{(2)}}\}_{k=1}^{j-1}]
            \leq 1/2.
        $ %
        This space has dimension $\leq d - j+1$.
        
        For any $t > 0$,
        define the bad event       
        $%
            B_j(t) := [g(\text{\small$\ket{v_j^{(2)}}$}) - 1/2 \geq t].
        $%
        Then we want to upper bound  
        $\Pr[ \intersect_{j=1}^{d/2} B_j(t) ] \leq 2\exp(-Cd^2t^2)$. 
        Toward this,
        \begin{align}
            \Pr[ \intersect_{j=1}^{d/2} B_j(t) ]
            = \prod_{j=1}^{d/2} \Pr[B_j(t) ~|~ B_1(t), ..., B_{j-1}(t)].
        \end{align}
        To upper bound each factor of the product, we further condition on 
        $\{\ket{v_k^{(2)}}\}_{k=1}^{j-1}$, and, using Lemma~\ref{lemma:lipschitz-sphere-concentration}, we have that
        \begin{align}
            \Pr[B_j(t) ~|~ B_1(t), ..., B_{j=1}(t)]
            \leq \exp(-ct^2(d-j)).
        \end{align}
        This implies that 
        $ %
            \Pr[ \intersect_{j=1}^{d/2} B_j(t) ]
            \leq \exp(-c \sum_{j=1}^{d/2} (d-j) t^2)
            = \exp(-t^2\Omega(d^2)).
        $ %
        This implies that with very high probability, there is some $j$ for which
        $g(\ket{v_j^{(2)}}) \leq 1/2$, which implies (\ref{expr:target-trace-difference}),
        which finally implies (\ref{expr:pairwise-diamond-probability}) and hence
        the claimed result.
\end{proof}

\begin{proof}[Proof of Proposition~\ref{prop:single-copy-lower-bound}]
    Suppose that there exists a $(\epsilon, \delta=\Theta(1))$-PAC learner for $\Hyp$.
    We will use it to construct a tester for the $\Packing(\gamma)$-hypothesis testing problem.  Suppose that $\sigma \in \Packing(\gamma)$ is the ground-truth
    hypothesis in the test.
    
    Given a dataset of size $m$, 
    we can use the learner to produce,
    with probability $\Omega(1)$, a measurement $h \in \Hyp$ satisfying
    $R(h) - R(\HelstM(\rho,\sigma)) \leq \epsilon$.  By
    Lemma~\ref{lemma:diamond-distance-risk-bound}, this implies
    $ %
        \epsilon 
        \geq \frac{\alpha}{4} \|h - \HelstM(\rho, \sigma)\|_{\diamond}^2
        = \frac{\gamma}{4d}\|h - \HelstM(\rho, \sigma)\|_{\diamond}^2,
    $ %
    which in turn implies that $\|h - \HelstM(\rho, \sigma)\|_{\diamond}^2$ can be
    made arbitrarily small by choosing $c$ sufficiently large in the
    stipulated upper bound on $\epsilon$. 

    Lemma~\ref{lemma:state-packing-diamond-distance} has two important implications:
    (i) the natural mapping from states $\sigma \in \Packing(\gamma)$ to measurements
    $\HelstM(\rho, \sigma) \in \HelstPacking(\gamma)$ is bijective;
    and 
    (ii) uniquely, $\HelstM(\rho, \sigma) = \argmin_{\hat{h} \in \HelstPacking(\gamma)} \|h - \hat{h}\|_{\diamond}$.
    To complete the test, we select $h_* := \argmin_{\hat{h} \in \HelstPacking(\gamma)} \|h - \hat{h}\|_{\diamond}$.  We then return the unique state in the state packing
    corresponding to $h_*$.  By (ii), this is guaranteed to be $\sigma$.
    Thus, an $(\epsilon, 1-\Omega(1))$-PAC learner with access to a dataset of size $m$
    can be used to solve the $\Packing(\gamma)$-hypothesis testing problem with probability $\Omega(1)$ for our choice of $\epsilon$.
    By the single-copy tomography lower bound, this implies that
    $m = \Omega(d^3/\epsilon^2) = \Omega(d^5)$, as desired.
\end{proof}

\subsection{Putting everything together}
\label{sec:separable-suboptimal-everything-together}
    To complete the proof of Theorem~\ref{thm:separable-learning-rules-suboptimal},
    we simply note that Propositions~\ref{prop:sample-complexity-upper-bound}
    and~\ref{prop:single-copy-lower-bound}, together with our choice of $\epsilon, \delta$ in Proposition~\ref{prop:single-copy-lower-bound}, imply that separable
    learning rules have strictly larger sample complexity than an optimal learning rule for the constructed learning scenario and ranges of parameters $\epsilon, \delta$.

\section{Proof of Theorem~\ref{thm:separable-non-separable-gap-bound}}
\label{proof:separable-non-separable-gap-bound}

Suppose that $\Hyp$ is $(\epsilon,\delta)$-PAC learnable
with optimal sample complexity
$m_0(\epsilon,\delta)$.
We first relate PAC learnability to the fat-shattering
dimension of $\Hyp$. By the classical lower bound for real-valued PAC learning
\cite[Thm.~19.5]{anthony1999neural},
there exists a universal constant $c>0$ such that
$%
    \fat_{c\epsilon}(\Hyp)
    \leq \mathrm{poly}
    \left(
        m_0(\epsilon,\delta),
        \frac{1}{\epsilon},
        \log\frac{1}{\delta}
    \right).
$ %
Here we have also used the fact that $m_0(\epsilon, \delta)$ is an upper bound for the classical sample complexity of learning the probability
function class induced by the hypothesis class.

Next, we invoke the sample complexity upper bound for denoised empirical risk minimization (DERM) established in Theorem~1 of \cite{Magner2024Fat}. In particular, equation~(9) therein states that DERM PAC learns $\Hyp$ using single-copy measurements with sample complexity upper bounded by 
$%
    m_{\mathrm{DERM}}(\epsilon,\delta) \le \mathrm{poly} \left( \fat_{c\epsilon}(\Hyp), N_{\mathrm{JM}}(\Hyp,\epsilon), \frac1\epsilon, \log\frac1\delta \right).
$ %
Substituting the fat-shattering upper bound into the \textsf{DERM} sample complexity expression yields
$%
    m_{\mathrm{DERM}}(\epsilon,\delta)
    \leq
    \mathrm{poly}\left(
        m_0(\epsilon,\delta),
        N_{\mathrm{JM}}(\Hyp,\epsilon),
        \frac{1}{\epsilon},
        \log\frac{1}{\delta}
    \right).
$ %

Therefore, the sample complexity of learning via
single-copy measurements is at most polynomially
worse in the listed parameters than the optimal unrestricted learning rule, which completes the proof.

\section*{Acknowledgments}
AI use disclosure: Google Gemini was used for validation of the proofs of auxiliary Lemmas~\ref{lemma:diamond-distance-risk-bound} and ~\ref{lemma:trace-distance-risk-bound} contained in the full paper.  ChatGPT was used for light feedback on grammar and spelling.  No AI was used to produce text in the paper.

\bibliographystyle{IEEEtran}
\bibliography{references}   %

\onecolumn
\section{Appendix}
\label{sec:appendix}

Here we give proofs of lemmas omitted in the main text.

\subsection{Proof of Lemma~\ref{lemma:diamond-distance-risk-bound}}

    For brevity, we denote by $\Delta R$ the excess risk to be lower bounded.  We also write $D := \HelstOp(\rho, \sigma)$.
    We can decompose $D$ into
    $D = D_+ - D_-$, the positive and negative parts of $D$.
    We denote by $H_0, H_1$ the projectors onto
    the eigenspaces of $D_+$ and $D_-$, respectively.
    We additionally denote the outcome-$0$ and outcome-$1$
    parts of the hypothesis $h$ by $M_0, M_1$, respectively.

    We can express $\Delta R$ as
    \begin{align}
        \Delta R 
        &= \Tr((D_+ - D_-) \cdot (H_0 - M_0)) \\
        &= \Tr(D_+ \cdot (\Id - M_0)) + \Tr(D_- \cdot M_0).
    \end{align}
    by simple algebra.

    Applying the lower bound on the singular values of $D$,
    we get
    \begin{align}
        \Delta R
        \geq \alpha \cdot [ \Tr( H_0 \cdot (I-M_0))
                        + \Tr(H_1 \cdot M_0)]
        = \alpha \cdot L,                
    \end{align}
    where we define
    \begin{align}
        L := \Tr(H_0) + \Tr(M_0) - 2 \Tr(H_0 M_0).
    \end{align}

    We will relate $L$ with the Frobenius norm
    $\|H_0 - M_0\|_{F}^2$.  We have
    \begin{align}
        \|H_0 - M_0\|_{F}^2
        &= \Tr((H_0 - M_0)^2) \\
        &= \Tr(H_0) - 2 \Tr(H_0M_0) + \Tr(M_0^2),
    \end{align}
    using cyclicity of the trace and the fact that $H_0$
    is a projection operator.

    Now
    \begin{align}
        L - \|H_0 - M_0\|_F^2
        = \Tr(M_0) - \Tr(M_0^2)
        = \Tr(M_0 \cdot (\Id - M_0)).
    \end{align}
    This can be seen to be non-negative because
    $0 \leq M_0 \leq \Id$, since $M_0$ is a POVM element.

    We have thus shown that $\Delta R \geq \alpha \cdot \|H_0 - M_0\|_F^2$.
    We can further lower bound the squared Frobenius norm by the squared spectral norm
    $\|H_0 - M_0\|_{op,L_2}^2$ (i.e., the largest eigenvalue of $H_0 - M_0$).  Finally, we have
    $\|H_0 - M_0\|_{op,L_2}^2 = 4 \|H - M\|_{\diamond}^2$,
    which completes the proof.

\subsection{Proof of Lemma~\ref{lemma:trace-distance-risk-bound}}

     For a POVM $h = (h_0, h_1)$, we can express $R(h)$
    as follows:
    \begin{align}
        &2\cdot R(h) = \Pr[ \Out(h, \rho) = 1] = \Pr[\Out(h, \sigma) = 0] \\
        &= 1 - \Tr((\rho - \sigma)\cdot h_0).
    \end{align}
    Let $\hat{h} := \HelstM(\rho, \hat{\sigma})$ and let
    $h_* := \HelstM(\rho, \sigma)$.
    Then we have
    \begin{align}
        2\cdot (R(\hat{h}) - R(h_*))
        = \Tr((\rho-\sigma) \cdot (h_{*,0} - \hat{h}_{0})).
    \end{align}
    Using optimality of $\hat{h}$ for distinguishing $\rho$ and
    $\hat{\sigma}$, we have
    \begin{align}
        \Tr(\rho \cdot (h_{*,0} - \hat{h}_0))
        \leq \Tr(\hat{\sigma} (h_{*,0} - \hat{h}_0)).
    \end{align}
    This implies 
    \begin{align}
        &\Tr((\rho-\sigma) \cdot (h_{*,0} - \hat{h}_{0})) \\
        &\leq \Tr( (\hat{\sigma} - \sigma)\cdot h_{*,0}) 
        - \Tr((\hat{\sigma} - \sigma) \cdot \hat{h}_0) \\
        & \max_{0 \leq P \leq \Id} \Tr((\hat{\sigma} - \sigma)P) + \max_{0 \leq P \leq \Id} \Tr((\sigma - \hat{\sigma})P)
        = \|\sigma - \hat{\sigma}\|_{1},
    \end{align}
    which completes the proof.

\end{document}